\documentclass[journal]{IEEEtran}
\usepackage{amsmath,amssymb}
\usepackage{subfigure}
\usepackage{graphicx,graphics,color,psfrag}
\usepackage{cite,balance}
\usepackage{caption}
\allowdisplaybreaks
\usepackage{algorithm}
\usepackage{algorithmic}
\usepackage{accents}
\usepackage{amsthm}
\usepackage{bm}
\usepackage{url}
\usepackage[english]{babel}
\usepackage{multirow}
\usepackage{enumerate}
\usepackage{cases}
\usepackage{stfloats}
\usepackage{dsfont}
\usepackage{color,soul}
\usepackage{amsfonts}
\usepackage{cite,graphicx,amsmath,amssymb}
\usepackage{subfigure}
\usepackage{fancyhdr}
\usepackage{hhline}
\usepackage{graphicx,graphics}
\usepackage{array,color}
\usepackage{mathtools}
\usepackage{amsmath}

\newtheorem{theorem}{\emph{\underline{Theorem}}}

\newtheorem{lemma}{\emph{\underline{Lemma}}}

\newtheorem{remark}{\bf \emph{\underline{Remark}}}

\def\({\left(}
\def\){\right)}

\def\b0{{\mathbf{0}}}

\newcommand{\nn}{\nonumber}

\usepackage[
top    = 0.58  in,
bottom = 0.58 in,
left   = 0.54 in,
right  = 0.54 in]{geometry}
\begin{document}
\captionsetup[figure]{name={Fig.}}     

\title{\huge 
Mixed Near- and Far-Field Communications for Extremely Large-Scale Array: An Interference Perspective}
\author{Yunpu Zhang, Changsheng You, \IEEEmembership{Member, IEEE}, Li Chen, \IEEEmembership{Senior Member, IEEE}, \\and Beixiong Zheng, \IEEEmembership{Member, IEEE}
\thanks{Y. Zhang and C. You are with the Department of Electronic and Electrical Engineering, Southern University of Science and Technology, Shenzhen 518055, China (e-mail: zhangyp2022@mail.sustech.edu.cn; youcs@sustech.edu.cn).
	
L. Chen is with the CAS Key Laboratory of Wireless-Optical Communications, University of Science and Technology of China, Hefei 230027, China (e-mail: chenli87@ustc.edu.cn).
	
B. Zheng is with the School of Microelectronics, South China University of Technology, Guangzhou 511442, China (e-mail: bxzheng@scut.edu.cn).
}\vspace{-0.5cm}}

\maketitle
\begin{abstract} 
Extremely large-scale array (XL-array) is envisioned to achieve super-high spectral efficiency in future wireless networks. Different from the existing works that mostly focus on the near-field communications, we consider in this paper a new and practical scenario, called \emph{mixed} near- and far-field communications, where there exist both near- and far-field users in the network. For this scenario, we first obtain a closed-form expression for the inter-user interference at the near-field user caused by the far-field beam by using Fresnel functions, based on which the effects of the number of BS antennas, far-field user (FU) angle, near-field user (NU) angle and distance are analyzed. We show that the strong interference exists when the number of the BS antennas and the NU distance are relatively small, and/or the NU and FU angle-difference is small. Then, we further obtain the achievable rate of the NU as well as its rate loss caused by the FU interference. Last, numerical results are provided to corroborate our analytical results. 
\end{abstract}
\vspace{-0.5cm}
\section{Introduction}




Extremely large-scale array/surface (XL-array/surface) has emerged as a promising technology to achieve the ever-increasing performance requirements of future sixth-generation (6G) wireless networks, such as super-high spectral efficiency and spatial resolution \cite{9903389,9326394}. This fundamentally leads to the communication paradigm shift from the conventional far-field communications (with planar-wave propagation) to the near-field communications (with spherical-wave propagation) \cite{9620081} and even the new \emph{mixed} near- and far-field communications (with both planar-/spherical-wave propagation). For example, consider an XL-array communication system where a base station (BS) equipped with an antenna
of diameter $0.5$ meter (m) communicates with users at $30$ GHz frequency. In this case, the well-known \emph{Rayleigh distance} is about $50$ m \cite{zhang2022fast}. As such, for a typical communication scenario, it may happen that some users reside in the near-field region, while the others locate in the far-field region, thus leading to several new design issues such as mixed-field channel estimation, joint near-/far-field beamforming, etc.

In particular, consider the inter-user interference in different communication scenarios. First, if the users considered are all in the far-field region, spatial division multiple access (SDMA) \cite{wu2022multiple} or beam division multiple access (BDMA) \cite{7093168} can be employed to simultaneously serve multiple users with low inter-user interference. This is because different far-field directional beams pointing towards different users are asymptotically orthogonal in the angular domain, thus eliminating the inter-user interference. Next, if the users locate in the near-field region, it has been recently shown in \cite{wu2022multiple,9738442} that the near-field users locating at different angles and/or (BS-user) distances can be served simultaneously by using the near-field beams with reduced inter-user interference. This is achieved by exploiting the unique near-field \emph{beam focusing} effect that enables the beam to be focused at a specific location rather than a specific direction in conventional far-field communications.
\begin{figure}[t]
	\centering
	\includegraphics[width=6cm]{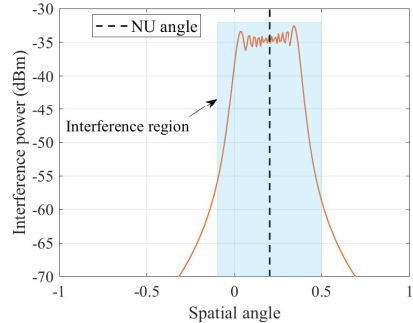}
	\caption{Interference power at a near-field user versus the spatial angle of a far-field beam, where the BS antenna number is 256, the BS transmit power is $30$ dBm, the carrier frequency is 30 GHz, and the BS-user distance is $3$ m.}
	\label{fig:Obs}
	\vspace{-20pt}
\end{figure}

{
Nevertheless, for the new mixed-field communication scenario, the inter-user interference analysis becomes more complicated. Specifically, the interference at the near-field user (NU) caused by the far-field user (FU) is fundamentally determined by the correlation between the NU channel steering vector and the FU beam. To examine it, we show in Fig.~\ref{fig:Obs} the interference power at a NU caused by the FU beam. An interesting observation is that the NU suffers strong interference from the FU beam, even when it locates at a spatial angle different from that of the FU (see the shadow area), which significantly differs from the results in the scenarios with NUs or FUs only.
This new phenomenon, however, has not been studied in the existing literature, which thus motivates the current work.

In this paper, we consider an XL-array communication system with one NU and one FU. First, we characterize the normalized interference power at the NU caused by the FU beam in closed form by using Fresnel functions,\footnote{The obtained results can be readily extended to analyze the interference at the FU caused by the NU beam.} based on which the effects of the number of BS antennas, FU angle, NU angle and distance are analyzed. It is shown that there is strong interference when the the number of BS antennas and the NU distance are relatively small, and/or the FU and NU angles are very close. Then, the resulting rate loss due to inter-user interference is obtained, and numerical results are provided to verify our theoretical results.}

\section{System Model}
As shown in Fig.~\ref{fig:SM}, we consider a two-user XL-array wireless communication system under the mixed near- and far-field communication scenario, where an $N$-antenna BS with a uniform linear array (ULA) serves a single-antenna NU and a single-antenna FU. Specifically, the NU and FU are respectively located near and far from the BS with the corresponding BS-user distance smaller and larger than the so-called Rayleigh distance, denoted by $Z=\frac{2D^2}{\lambda}$ with $D$ and $\lambda$ representing the array aperture and carrier wavelength, respectively.\footnote{The NU reduces to a FU when $N$ is sufficiently small and/or the BS-user distance is sufficiently large.}

\vspace{-12pt}
\subsection{Channel Models}
\underline{\bf 1) Far-field channel model:} Similar to \cite{9129778}, when the user locates sufficiently far from the BS (in the far-field region), the BS$\to$FU channel can be characterized based on the planar-wave assumption, which is given by
\begin{equation}
	 \mathbf{h}^H_{\rm far}=\sqrt{N} h_{\rm far}\mathbf{a}^{H}(\psi),
	 	\vspace{-3pt}
\end{equation}
where $h_{\rm far}$ denotes the complex-valued channel gain between the BS and FU. $\mathbf{a}^{H}(\psi)$
denotes the far-field steering vector, given by
\begin{equation}\label{far_steering}
    \mathbf{a}(\psi)\triangleq \frac{1}{\sqrt{N}}\left[1, e^{j \pi \psi},\cdots, e^{j \pi (N-1)\psi}\right]^{T},
    	\vspace{-3pt}
\end{equation}
where $\psi=2d\cos(\varphi)/ \lambda$ denotes the spatial angle of the FU with $\varphi$ being the physical angle-of-departure (AoD) from the BS center to the FU. Without loss of generality, we assume that the $N$-antenna BS is placed along the $y$-axis and the $n$-th antenna is located at $(0,\delta_{n}d)$ m, where $\delta_{n}=\frac{2n-N+1}{2}$ with $n=0,1,\ldots,N-1$, and $d=\frac{\lambda}{2}$ denotes the antenna spacing. 
\begin{figure}[t]
	\centering
	\includegraphics[width=8cm]{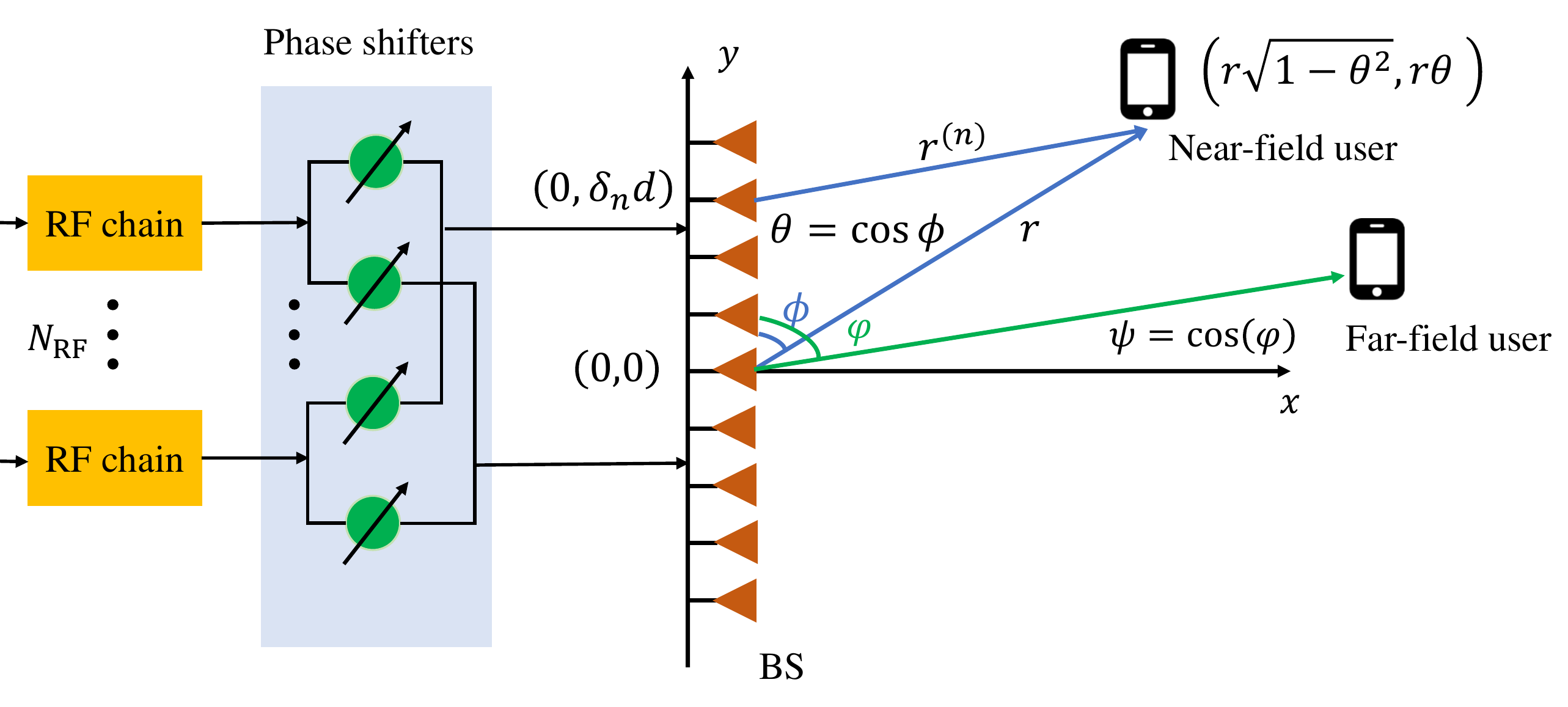}
	\caption{A two-user XL-array wireless communication system.}
	\label{fig:SM}
	\vspace{-15pt}
\end{figure}

\underline{\bf 2) Near-field channel model:} For the NU, we consider the more accurate spherical-wave propagation model, which applies as well when the user locates in the far-field region. As such, the BS$\to$NU channel can be modeled as
\begin{equation}\label{Eq:nf-model}
    \mathbf{h}^H_{\rm near}=\sqrt{N}h_{\rm near} \mathbf{b}^{H}(\theta,r),
\end{equation}
where $h_{\rm near}=\frac{\sqrt{\beta}}{r}  e^{-\frac{\jmath 2 \pi r}{\lambda}}$ is the complex-valued channel gain\footnote{{In this paper, we consider the Fresnel region with the BS-NU distance $r$ larger than $1.2D$, for which the amplitude variations over XL-array antennas are negligible \cite{9723331}, while the case with non-negligible amplitude variations will be left for future work.}} with $\beta$ and $r$ denoting the reference channel gain at a distance of $1$ m and the distance between the BS center and the NU, respectively. $\mathbf{b}^{H}(\theta,r)$ denotes the near-field steering vector, which is given by \cite{luo2022beam}
\begin{equation}\label{near_steering}
    \mathbf{b}\left(\theta, r\right)\!=\!\frac{1}{\sqrt{N}}\!\!\left[e^{-j 2 \pi(r^{(0)}-r)/\lambda}, \cdots, e^{-j 2 \pi(r^{(N-1)}-r)/\lambda}\right]^{T},
\end{equation}
where $\theta=2d\cos(\phi)/ \lambda$ denotes the spatial angle of the NU with $\phi$ denoting the physical AoD from the BS center to the NU; $r^{(n)}=\sqrt{r^2+\delta_{n}^2d^2-2r\theta \delta_{n}d}$ represents the distance between the $n$-th antenna at the BS (i.e., ($0, \delta_{n}d$)) and NU.
\vspace{-10pt}
\subsection{Signal Model under Mixed-Field Communications}
The BS is equipped with $N_{\rm RF}\ge K$ radio frequency (RF) chains to enable multiuser communications (i.e., $K$ users). Without loss of generality, we assume $N_{\rm RF}=2$ for the two-user case. Moreover, to facilitate the interference analysis in the sequel, we assume the analog-only beamforming for each user, while the analysis can be extended when digital beamforming is applied to further reduce the interference. Let $x_{k}$, $k=\{1,2\}$ denote the transmitted symbol by the BS to user $k$ with power $P_{k}$, and $\mathbf{v}_{\rm near}$ and $\mathbf{v}_{\rm far}$ represent the transmit beamformers for the NU and FU, respectively. Then, 
the received signal at the NU is given by
\begin{align}
    y_{\rm near}&=\mathbf{h}^H_{\rm near}\mathbf{v}_{\rm near}x_{\rm near}+\mathbf{h}^H_{\rm near}\mathbf{v}_{\rm far}x_{\rm far}+z_{0}\nn\\
    &=\sqrt{N}h_{\rm near}\mathbf{b}^{H}(\theta,r)(\mathbf{v}_{\rm near}x_{\rm near}+\mathbf{v}_{\rm far}x_{\rm far})+z_{0},
\end{align}
where $z_{0}$ is the received additive white Gaussian noise (AWGN) at the NU with power $\sigma^2$.

As such, the receive signal-to-interference-plus-noise ratio (SINR) at the NU is given by
\begin{equation}\label{Eq:real_rete}
	\mathbf{\rm SINR}_{\rm near}=\frac{P_{\rm near}\left|\sqrt{N} h_{\rm near} \mathbf{b}^H\left(\theta, r\right) \mathbf{v}_{\rm near}\right|^2}{P_{\rm far}\left|\sqrt{N} h_{\rm near} \mathbf{b}^H\left(\theta, r\right) \mathbf{v}_{\rm far}\right|^2+\sigma^2}.
\end{equation}

\noindent For ease of implementation and maximizing the received signal power at the intended user, the beamformers for these two users are respectively designed as $\mathbf{v}_{\rm near}=\mathbf{b}(\theta, r)$ and $\mathbf{v}_{\rm far}=\mathbf{a}(\psi)$. Hence, the achievable rate of the NU in bits per second per Hertz (bps/Hz) is given by 
\begin{align}\label{Eq:real_rete1}
	R_{\rm near}&=\log _2\left(1+	\mathbf{\rm SINR}_{\rm near}\right)\nn\\
	&=\log _2\left(1+\frac{P_{\rm near}g_{\rm near}}{P_{\rm far}g_{\rm near}f^2(N,\psi,\theta,r)+\sigma^2}\right),
\end{align}
where $g_{\rm near}=N\beta/r^2$, and  $f(N,\psi,\theta,r)=\left|\mathbf{b}^H\left(\theta, r\right) \mathbf{a}(\psi)\right|$ is defined as the \emph{normalized (mixed-field) interference power} caused by the far-field beam to the near-field channel steering vector. In the following, we first characterize useful properties of the normalized interference power and then obtain the achievable rate of the NU.
\vspace{-6pt}
\section{Normalized Interference Power Analysis}
\subsection{Normalized Interference Power}
First, with the definitions of $\mathbf{a}(\psi)$ in \eqref{far_steering} and $\mathbf{b}^H\left(\theta, r\right)$ in \eqref{near_steering}, the normalized interference power $f(N,\psi,\theta,r)$ can be explicitly expressed as
		\begin{align}\label{Eq:New_f}
			f&(N,\psi,\theta,r)\triangleq\left|\mathbf{b}^H\left(\theta, r\right) \mathbf{a}(\psi)\right|\nn\\
			&\stackrel{(a_1)}{\approx}\frac{1}{N}\left|\sum_{\delta_{n}} e^{j \frac{2\pi}{\lambda}\left(-\delta_{n}d \theta+\frac{ \delta_{n}^2 d^2\left(1-\theta^2\right)}{2r}\right)+j \pi\left((\delta_{n}+\frac{N-1}{2}) \psi\right)}\right|\nn\\
			&\stackrel{(a_2)}{=}\frac{1}{N}\left|\sum_{n=0}^{N-1} e^{j \pi\left(\frac{ n^2 d(1-\theta^2)}{2r}-n\left(\theta-\psi+\frac{d(N-1)(1-\theta^2)}{2r}\right) \right)}\right|,
		\end{align}
	where $(a_1)$ follows from the taylor expansion and Fresnel approximation with $r^{(n)}=\sqrt{r^2+\delta_{n}^2d^2-2r\theta \delta_{n}d} \approx r-\delta_{n}d\theta+\frac{\delta_{n}^2d^2(1-\theta^2)}{2r}$, and $(a_2)$ is obtained by replacing $\delta_{n}$ with $n$, i.e., $\delta_{n}=\frac{2n-N+1}{2}$. {Note that $(a_1)$ is accurate when the BS-NU/scatter distance is larger than $0.5\sqrt{D^3/\lambda}$, which is much smaller than the Rayleigh distance $2D^2/\lambda$.}

However, the normalized interference power in \eqref{Eq:New_f} is still in a complicated form, thus making it hard to characterize the properties of normalized interference power. To tackle this issue, similar to \cite{deshpande2022wideband}, we approximate the normalized interference power in a more tractable form by using Fresnel functions.

\begin{theorem}\label{Th1}
	\emph{
		The normalized interference power $f(N,\psi,\theta,r)$ in \eqref{Eq:New_f} can be approximated as
		\begin{equation}\label{Eq:Th1}
			f(N,\psi,\theta,r)\approx
				G(\beta_{1},\beta_{2})= \left|\frac{\hat{C}(\beta_{1},\beta_{2})+j\hat{S}(\beta_{1},\beta_{2})}{2\beta_{2}}\right|,		
		\end{equation}
	where
	\begin{equation}\label{Eq:beta1}
		\beta_{1}=(\theta-\psi)\sqrt{\frac{r}{d(1-\theta^2)}}, 
	\end{equation}
	\begin{equation}\label{Eq:beta2}
	\beta_{2}=\frac{N}{2}\sqrt{\frac{d(1-\theta^2)}{r}}.
	\end{equation}
Moreover, $\hat{C}(\beta_{1},\beta_{2})\triangleq C(\beta_{1}+\beta_{2})-C(\beta_{1}-\beta_{2})$ and $\hat{S}(\beta_{1},\beta_{2})\triangleq S(\beta_{1}+\beta_{2})-S(\beta_{1}-\beta_{2})$, where $C(\cdot)$ and $S(\cdot)$ are the Fresnel integrals, defined as  $C(x)=\int_{0}^{x}\cos(\frac{\pi}{2}t^2) {\rm d} t$ and $S(x)=\int_{0}^{x}\sin(\frac{\pi}{2}t^2) {\rm d} t$. }
\end{theorem}
\begin{proof}
	First, from \eqref{Eq:New_f}, we have
		\begin{align}\label{Eq:good}
		f(N,\psi,\theta,r)=\left|\frac{1}{N}\sum_{n=0}^{N-1} e^{j \pi\left(A_{1}n-A_{2}\right)^2}\right|
		=\left| F(A_{1},A_{2})\right|,\nn
	\end{align}
	where  $A_{1}=\sqrt{\frac{ d\left(1-\theta^2\right)}{2r}}$ and $A_{2}=\frac{1}{A_{1}}\left( \frac{\theta-\psi}{2}+\frac{(N-1)d(1-\theta^2)}{4r}\right)$.
\begin{figure}[t]
	\centering
	\subfigure[{3D illustration}]{\includegraphics[width=4.7cm]{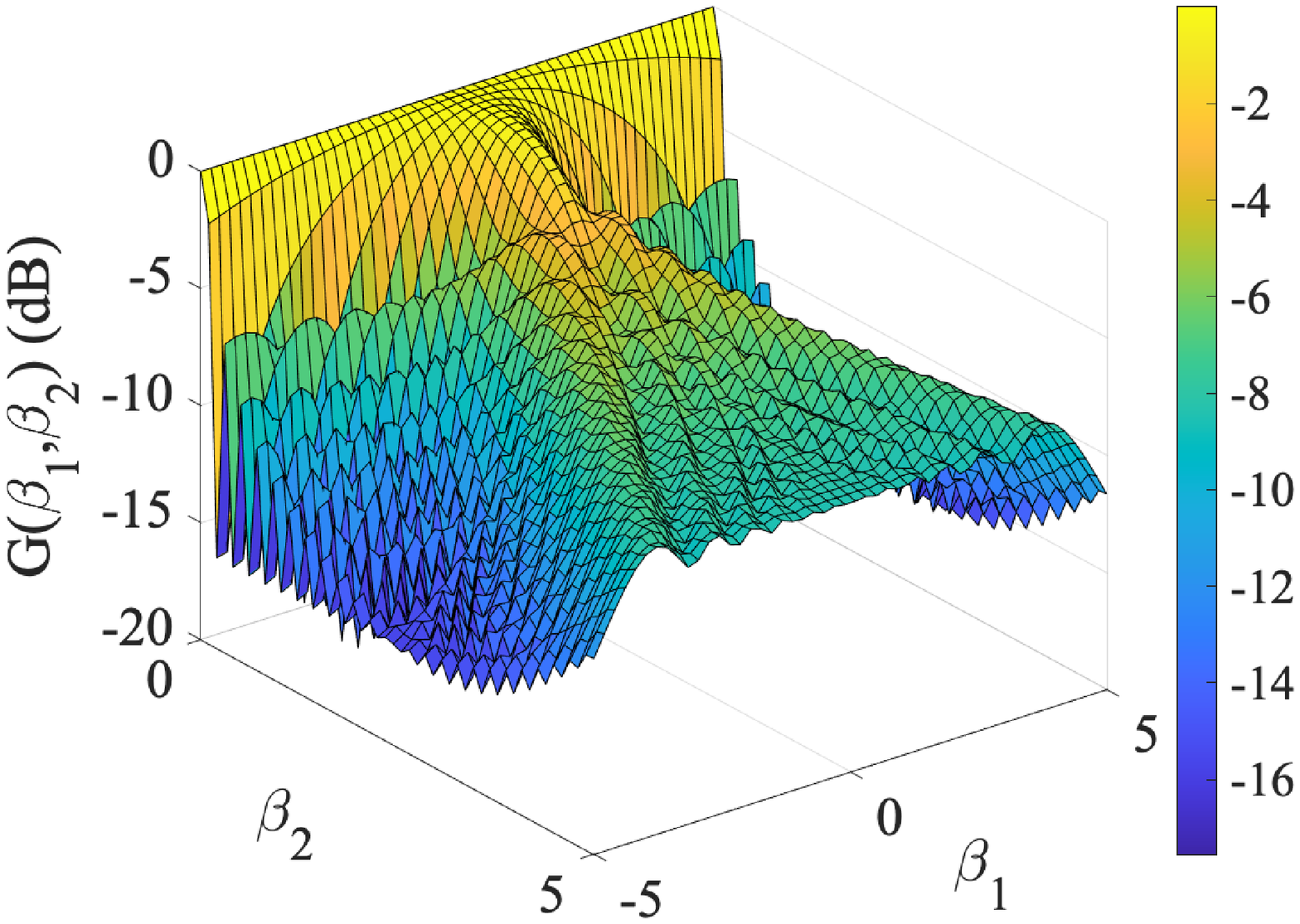}} 	\hspace{-12pt}
 \subfigure[{2D illustration}]{\includegraphics[width=4.7cm]{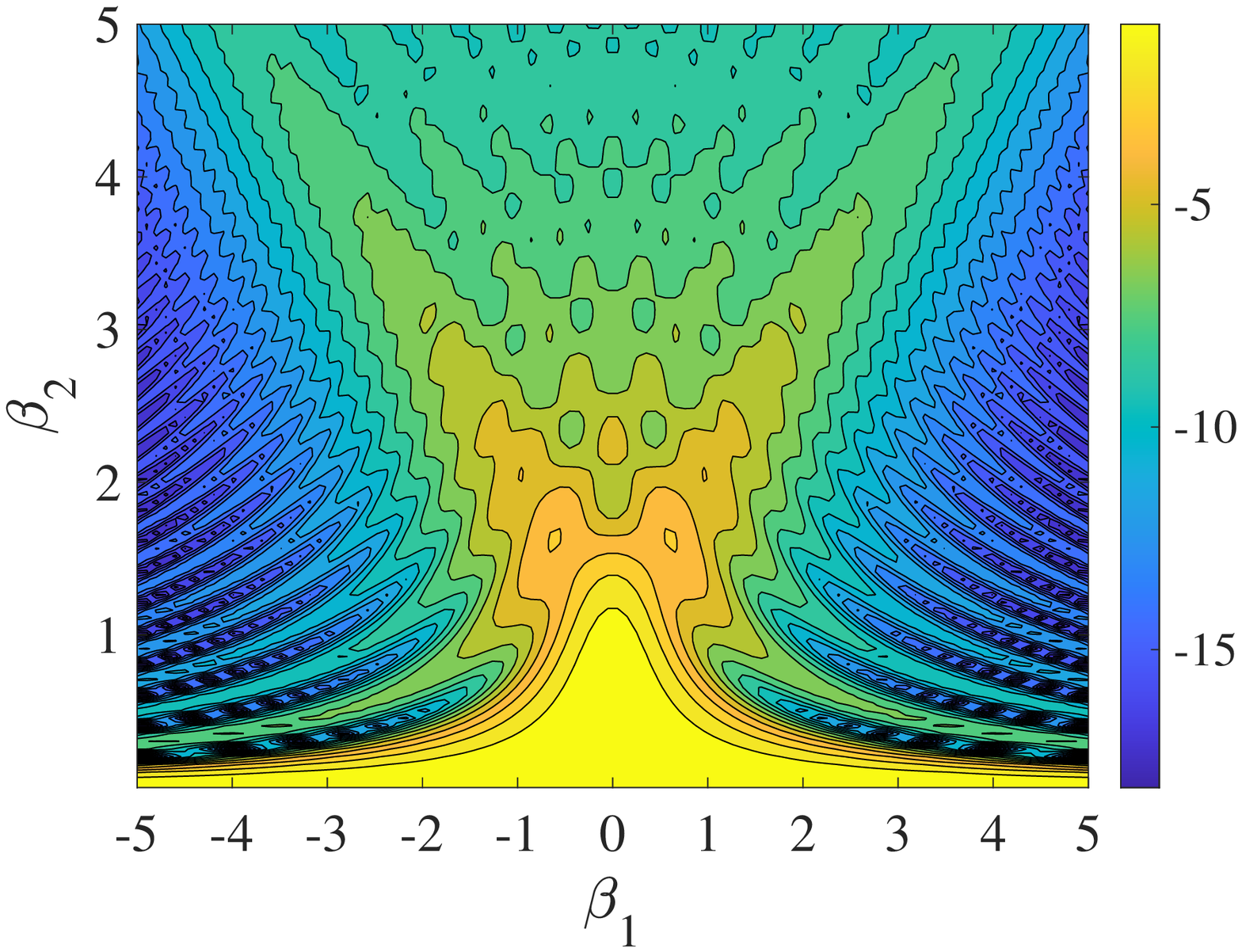}\label{fig:2D}}
	\caption{Illustrations of function $G(\beta_{1},\beta_{2})$ in (9) versus $\beta_{1}$ and $\beta_{2}$.}
	\label{fig:2_3D}
	\vspace{-8pt}
\end{figure}	
	Then, the summation in $F(A_{1}, A_{2})$ can be approximated as an integral,
	\begin{align}
		F(A_{1},A_{2})&\stackrel{(b_1)}{\approx}\frac{1}{N}\int_{0}^{N} e^{j \pi\left(A_{1}n-A_{2}\right)^2} {\rm d} n \nn\\
		&\stackrel{(b_2)}{=}\frac{1}{N\sqrt{2}A_{1}}\int_{-\sqrt{2}A_{2}}^{\sqrt{2}A_{1}N-\sqrt{2}A_{2}} e^{j \pi \frac{1}{2}t^2} {\rm d} t,
	\end{align}
	where $(b_1)$ is accurate when $N\to\infty$ according to the Riemann integral, and $(b_2)$ is obtained by letting $A_{1}n-A_{2}=\frac{\sqrt{2}}{2}t$.
 
	Next, based on the Fresnel integrals, we have 
	\begin{align}\label{Eq:FFF}
			&F(A_{1},A_{2})=\frac{1}{N\sqrt{2}A_{1}}\int_{-\sqrt{2}A_{2}}^{\sqrt{2}A_{1}N-\sqrt{2}A_{2}} e^{j \pi \frac{1}{2}t^2} {\rm d} t\nn\\
			&=\frac{\int_{0}^{\sqrt{2}A_{1}N-\sqrt{2}A_{2}} e^{ j \pi\frac{1}{2}t^2} {\rm d} t-\int_{0}^{-\sqrt{2}A_{2}} e^{j \pi\frac{1}{2}t^2} {\rm d} t}{\sqrt{2}A_{1}N}\nn\\
			&=\frac{C(\beta_{1}+\beta_{2})-C(\beta_{1}-\beta_{2})+j\left(S(\beta_{1}+\beta_{2})-S(\beta_{1}-\beta_{2})\right)}{2\beta_{2}},\nn
	\end{align}
	where $\beta_{1}=(\theta-\psi)\sqrt{\frac{r}{d(1-\theta^2)}}$ and $\beta_{2}=\frac{N}{2}\sqrt{\frac{d(1-\theta^2)}{r}}$.
	
	By defining $G(\beta_{1},\beta_{2})=\left| F(A_{1},A_{2})\right|$ and combining the above leads to the desired result in \eqref{Eq:Th1}.
\end{proof}
	\vspace{-6pt}
\begin{remark}[Useful properties of function $G(\cdot)$]\label{Re:Property}
    \emph{To draw useful insights, we first illustrate in Fig.~\ref{fig:2_3D} that function $G(\beta_{1},\beta_{2})$ versus $\beta_1$ and $\beta_2$ in both 2D and 3D. Several important properties of the function $G(\cdot)$ are summarized as follows.
    \begin{itemize}
        \item First, $G(\cdot)$ is symmetric with respect to $\beta_{1}$. With $\beta_{1}$ defined in \eqref{Eq:beta1}, this symmetry indicates that the normalized interference power keeps unchanged when the absolute values of the near-and-far user angle-difference $\theta-\psi$ are the same. Moreover, it is worth mentioning that $G(\cdot)$ is almost invariable with $\beta_{1}$ when $\beta_{2}$ is sufficiently small.
        \item Second, given $\beta_{1}$, $G(\cdot)$ generally decreases with $\beta_2$, while there exist large and small fluctuations in $G(\cdot)$ with growing $\beta_2$, when $|\beta_{1}|$ is relatively large (e.g., $|\beta_{1}|>0.6$) and small (e.g., $|\beta_{1}|\le0.6$).
        \item Moreover, it is worth noting that the derived closed-form expression for $G(\cdot)$ in \eqref{Eq:Th1} is a generalization of the column coherence between two near-field steering vectors defined in \cite{9693928} by setting $\beta_{1}=0$, i.e., both near- and far-field steering vectors are at the same angle ($\theta=\psi$).
    \end{itemize}}
\end{remark}

  \begin{figure}[t]
 	\centering
 	\subfigure[{Accuracy of approximation in \eqref{Eq:Th1}. }]{\includegraphics[width=4.8cm]{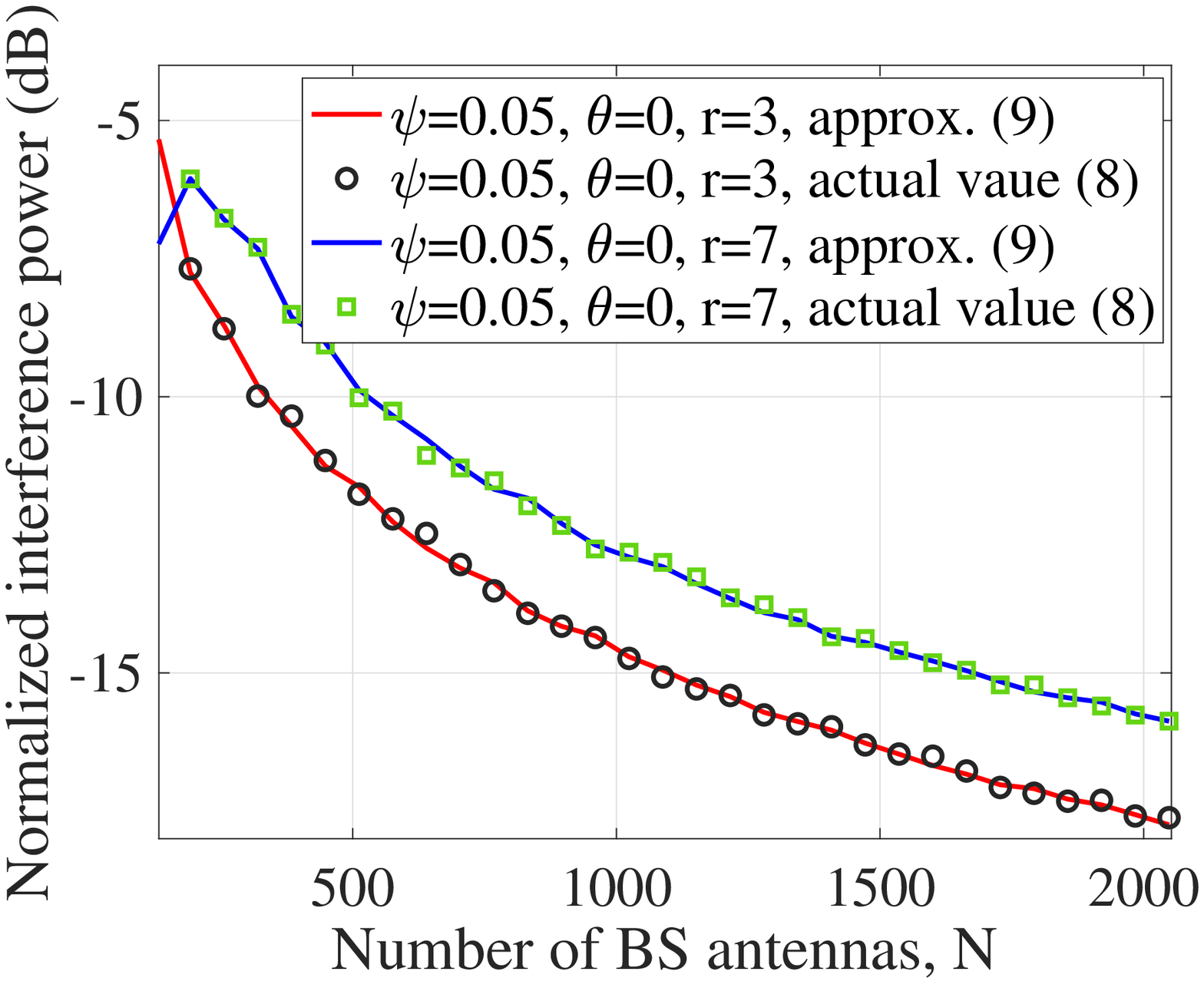}\label{fig:verification}}
 	\hspace{-18pt}
 	\subfigure[{Effect of number of BS antennas.}]{\includegraphics[width=4.8cm]{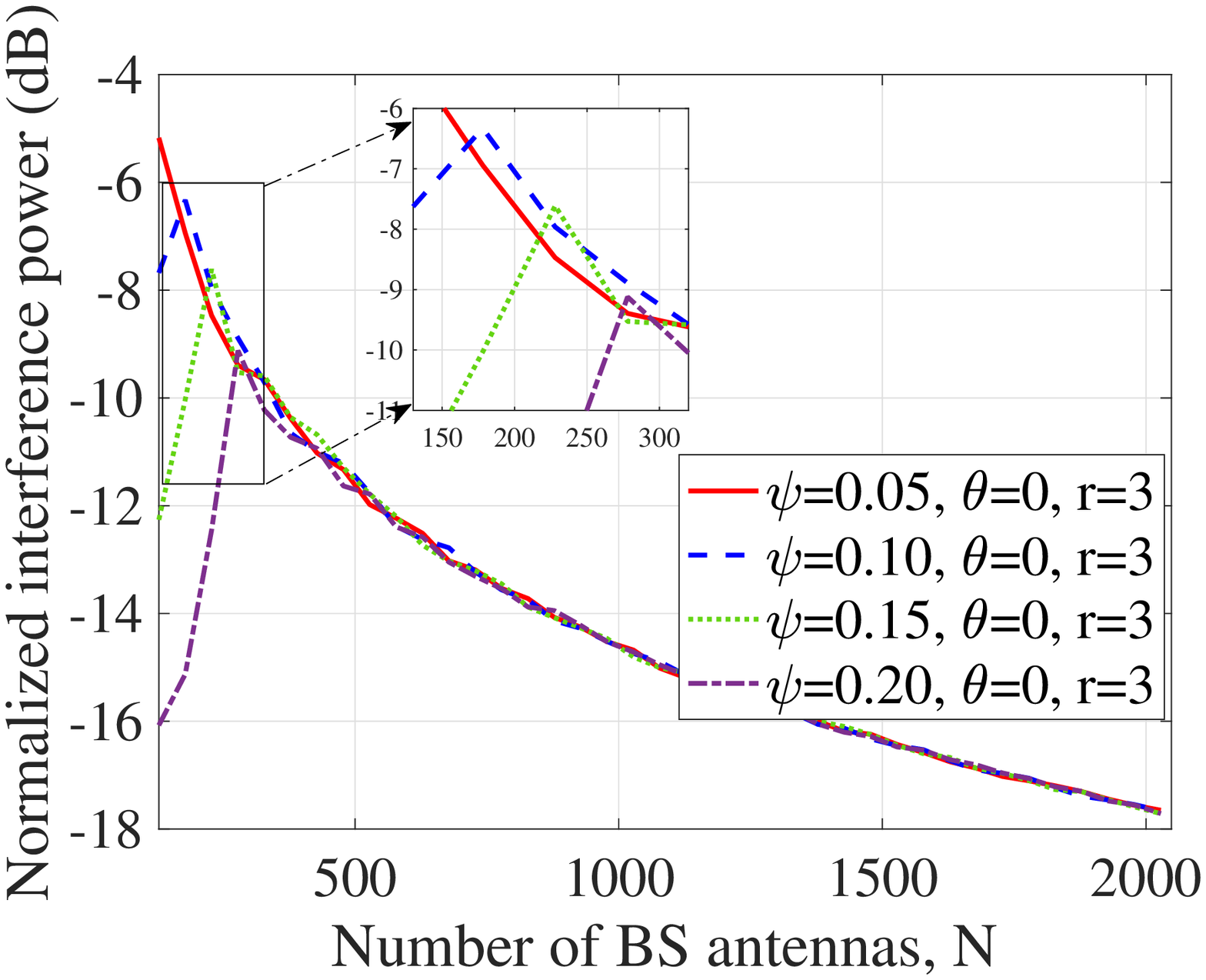}\label{figg:	BSAntennax}}
 	\caption{{Normalized interference power vs. number of BS antennas.}}
 	\vspace{-8pt}
 \end{figure}
	\vspace{-8pt}
\begin{remark}[What determines the normalized interference power?]
\emph{Theorem~\ref{Th1} shows an important result that the normalized interference power is fundamentally determined by the function $G(\cdot)$, as well as the two parameters $\beta_{1}$ and $\beta_{2}$. To be more specific, $\beta_{1}$ is a function of the FU (spatial) angle, the NU angle and distance, while $\beta_{2}$ is jointly determined by the number of BS antennas, as well as the NU angle and distance. These factors will be studied in more details in the next.}
\end{remark}
\vspace{-14pt}
\subsection{Effects of Key Parameters}
\subsubsection{Effect of the Number of BS Antennas}
{To begin with, we first plot Fig.~\ref{fig:verification} to compare the obtained approximated normalized interference power in \eqref{Eq:Th1} and its actual value in \eqref{Eq:New_f}. It is observed that the approximation in \eqref{Eq:Th1} is accurate under different numbers of BS antennas and NU distances.} Next,
given the fixed FU angle, NU angle and distance, $\beta_{2}$ is affected by the number of BS antennas $N$ only, while $\beta_{1}$ becomes a constant. Moreover, the BS antenna size affects the near-field region, since the Rayleigh distance, $Z=\frac{1}{2}N^2\lambda$, is quadratically increasing with the number of BS antennas. To characterize the effect of $N$, we plot in Fig.~\ref{figg:	BSAntennax} the normalized interference power versus the number of BS antennas. First, it is observed that the normalized interference power is first increasing and then decreasing with $N$. The reason is that, {when $N$ is small, increasing $N$ leads to a wider interference region in the angular domain (see Fig.~\ref{fig:Obs}); hence, increasing $N$ renders the NU more likely reside in the interference region and thus growing interference power.} Nevertheless, when $N$ is sufficiently large, the NU (in the interference region) suffers decreasing interference power with an increasing $N$ due to the wider interference region. Second, when $N$ is relatively small (i.e., $\beta_{2}$ is small), a larger angle-difference (i.e., larger $\beta_{1}$) leads to a smaller interference, which is consistent with the second property in Remark~\ref{Re:Property}. {By contrast, in the large-$N$ regime (e.g., $N>400$), the interference powers for different angle differences are comparable. This is expected since the interference region becomes wider with a larger $N$ and will occupy almost the whole spatial region when $N$ is sufficiently large, for which different angle-differences will suffer comparable interference power.}
  \begin{figure*}[t]
	\centering
		\centering
	\subfigure[Effect of FU angle.]{\includegraphics[width=5cm]{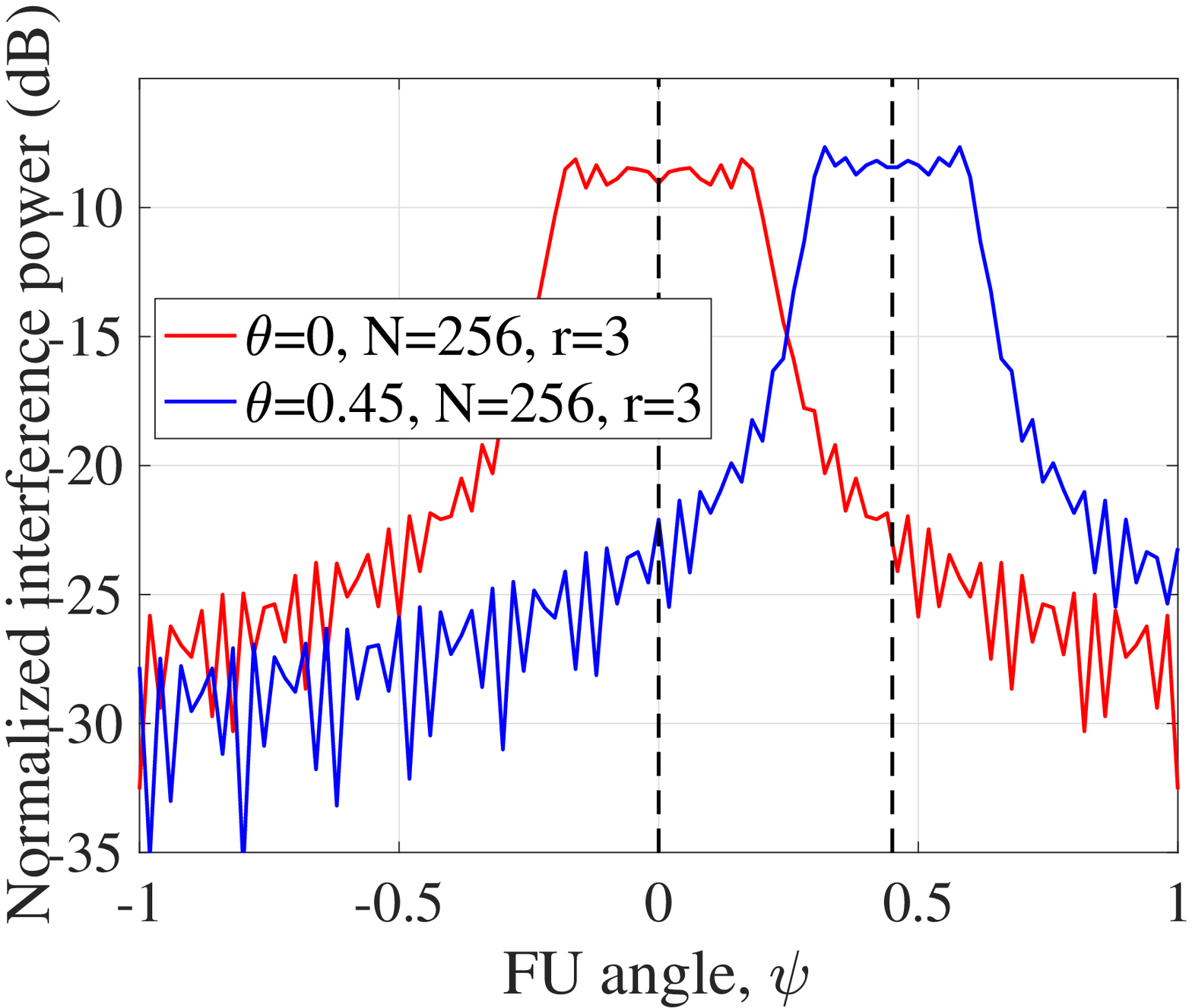}\label{fig:FU_angle}}
		\hspace{-18pt}
	\subfigure[Effect of angle-difference.]{\includegraphics[width=5cm]{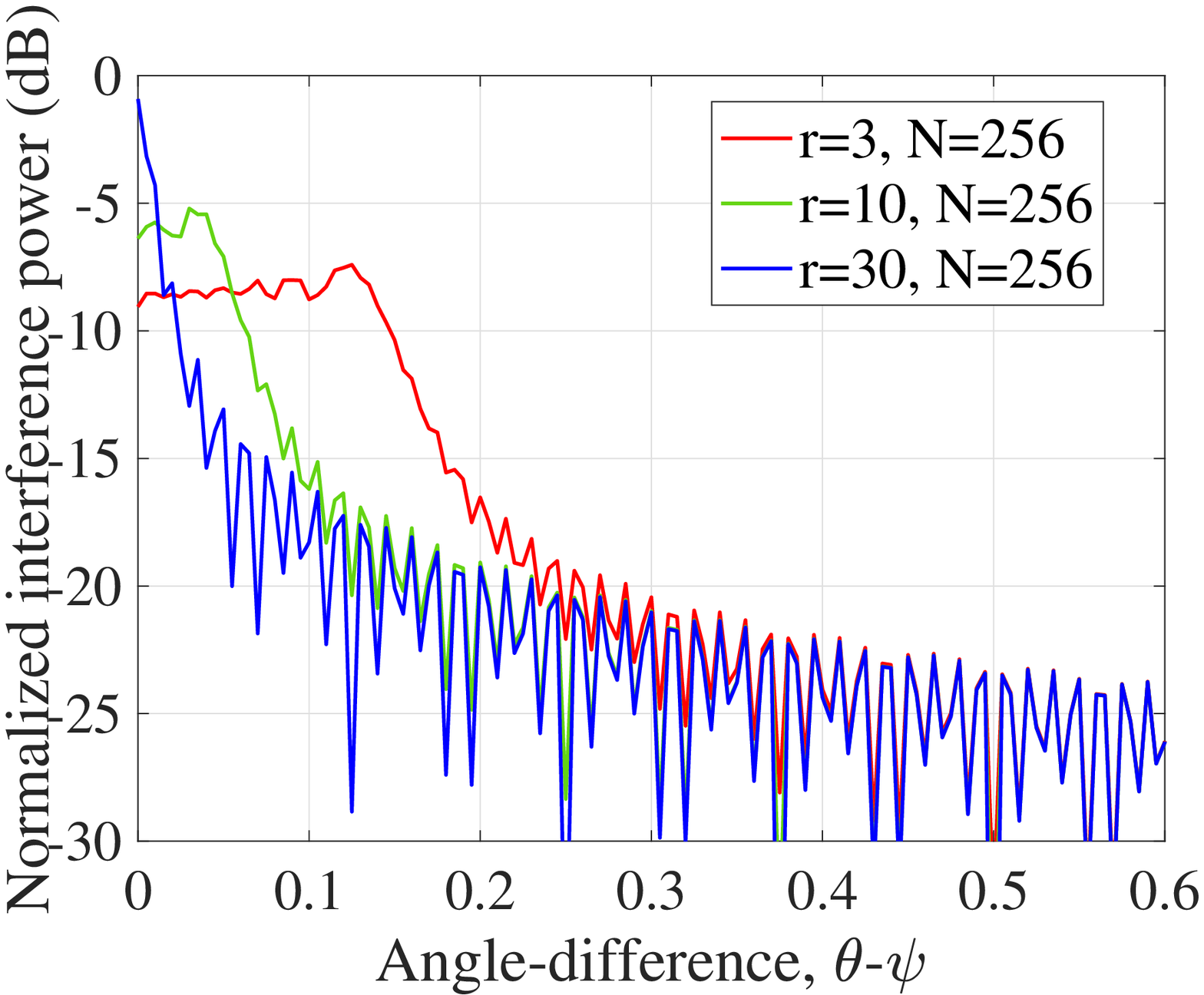}\label{fig:Effect_Angle_Diff}}
		\hspace{-18pt}
	\subfigure[Effect of NU angle.]{\includegraphics[width=5cm]{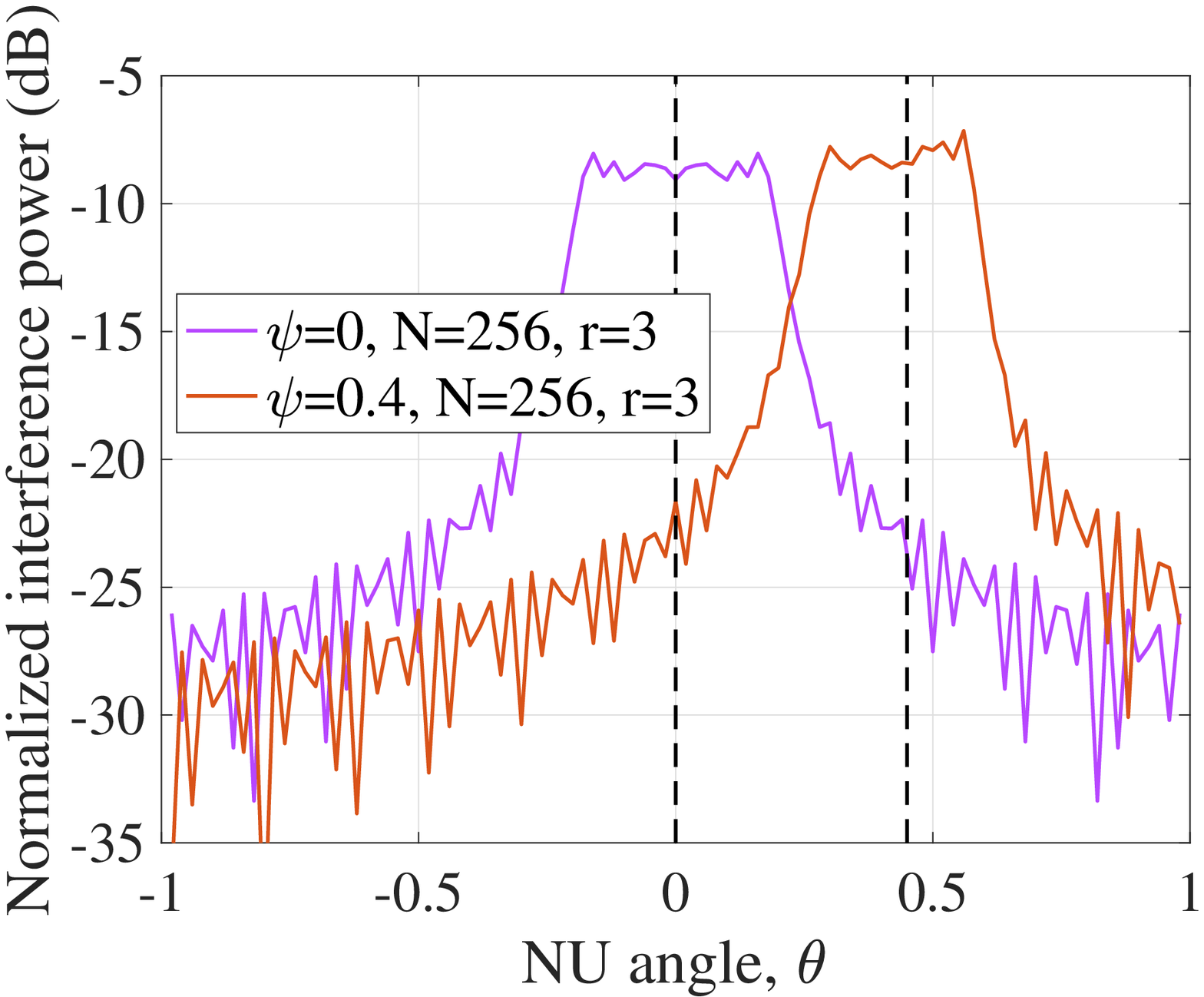}\label{fig:NU_angle}}
		\hspace{-18pt}
	\subfigure[Effect of NU distance.]{\includegraphics[width=5cm]{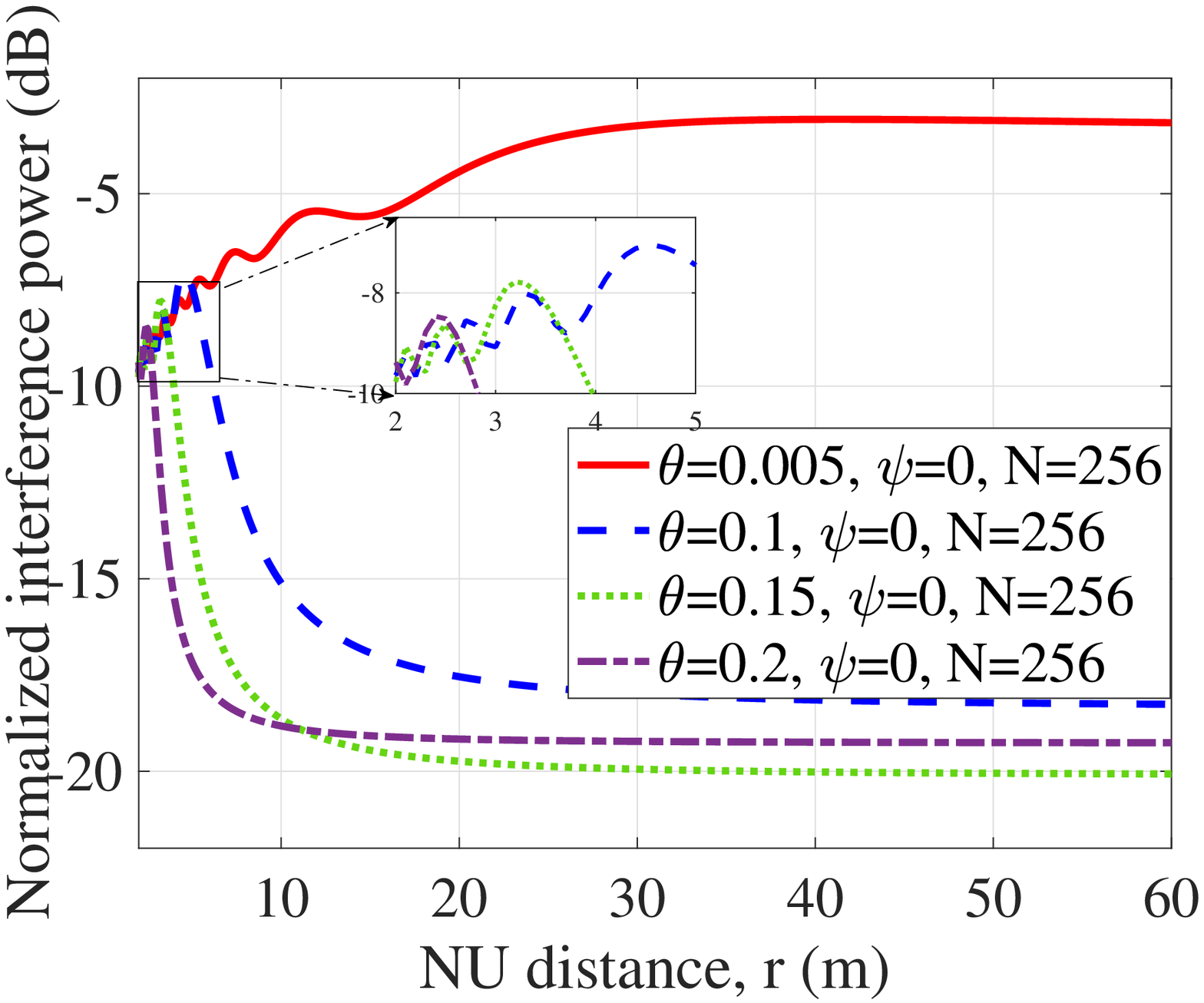}\label{fig:Effect_NU_Dist}}
	\caption{Normalized interference power versus key parameters, where the normalized interference power is obtained based on \eqref{Eq:Th1}.}
	\vspace{-16pt}
\end{figure*}

\subsubsection{Effect of the FU Angle}
Next, given the fixed number of BS antennas, the NU angle and distance, we characterize the effect of FU angle. In this case, $\beta_{1}$ is determined by the FU angle $\psi$, while $\beta_{2}$ is a constant. To be more specific, we plot in Fig.~\ref{fig:FU_angle} the normalized interference power versus the FU angle. An interesting observation is that when the FU beam angle locates in the neighborhood of the NU angle, there is always strong interference power at the NU. Moreover, the interference power is symmetric with respect to the NU angle, which is in accordance with the first property in Remark~\ref{Re:Property}. 

Given the NU angle, the FU angle also determines the near-and-far user angle-difference, whose effects are discussed below.
\vspace{-8pt}
\begin{remark}[Near-and-far user angle-difference]\label{Re:Angle_diff}
	\emph{In Fig.~\ref{fig:Effect_Angle_Diff}, we show the effect of the angle-difference on the normalized interference power. First, it is observed that, when the angle-difference is sufficiently small, e.g., $\theta-\psi<0.01$, a longer NU distance (e.g., $r=30$ m) yields a higher interference
		power. This is intuitively expected since when the NU and FU are very close, the interference power will increase with the NU distance, and the
		maximum interference power is obtained when the
		NU distance exceeds the Rayleigh distance (i.e., far-field region). Second, as the angle-difference increases, the interference power for different angle-differences generally decreases, and eventually appears very close when the angle-difference exceeds a threshold (e.g., $\theta-\psi>0.35$). Besides, the interference power with a longer NU distance decreases faster. This reason is that a longer NU distance leads to a narrower interference region.
	}
\end{remark}

\subsubsection{Effect of the NU Angle}
Note that both $\beta_{1}$ and $\beta_{2}$ are functions of NU angle $\theta$ with the fixed number of BS antennas, FU angle, and NU distance, thus making it difficult to obtain clear insights with respect to the effect of the NU angle. To this end,
we plot in Fig.~\ref{fig:NU_angle} the normalized interference power versus the NU angle $\theta$. It is observed that the effect of NU angle is very similar to that of the FU angle, whereas in the large-$\theta$ regime, the interference power fluctuates more drastically. 

\subsubsection{Effect of the NU Distance}
Last, we study the effect of the NU distance $r$ on the normalized interference power, given the fixed number of BS antennas, NU and FU angles. Based on Theorem~\ref{Th1}, it can be easily verified that $\beta_{1}$ monotonically increases with $r$, while $\beta_{2}$ decreases with $r$. Although it is difficult to analytically characterize the effect of NU distance, we provide the numerical result in Fig.~\ref{fig:Effect_NU_Dist} for illustration. Specifically, similar to Fig.~\ref{figg:	BSAntennax}, for the case with a small angle-difference (e.g., $\theta-\psi=0.1, 0.15$ and $0.2$), the normalized interference power first slightly increases and then drastically decreases with the NU distance. In contrast, when the angle-difference is sufficiently small (e.g., $\theta-\psi=0.005$), the normalized interference power first increases and then saturates when $r$ is sufficiently large, which is in accordance with Remark~\ref{Re:Angle_diff}. Moreover, it is observed that, with a very small angle-difference (e.g., $\theta-\psi=0.005$), the NU suffers the strongest interference power when it is sufficiently far from the BS, for which case it reduces to a FU.

\vspace{-12pt}
\section{Rate Loss Analysis}
In this section, we study the effect of normalized interference power on the achievable rate of the NU. For notational brevity, we use $f$ to denote $f(N,\psi,\theta,r)$ in the sequel.
\vspace{-12pt}
\subsection{Rate Loss}
To characterize the rate loss due to the FU interference, we first provide the ideal achievable rate of the NU with no interference, which is given by
\begin{equation}\label{Ideal_rate}
	R^*_{\rm near}=\log _2\left(1+\frac{P_{\rm near}g_{\rm near}}{\sigma^2}\right),
	\vspace{-3pt}
\end{equation}
Then, we define $\Delta_{R}=R^*_{\rm near}-R_{\rm near}$ as the rate loss caused by the FU interference, which is upper-bounded as follows.
\begin{lemma}\label{Le:Rate_gap}
    \emph{The rate performance loss $\Delta_{R}$ can be upper-bounded as
    \begin{equation}\label{Eq:loss}
			\Delta_{R}\le\log _2\left(1+\frac{P_{\rm near}g_{\rm near}}{\sigma^2}\cdot\frac{P_{\rm far}f^2}{P_{\rm far}f^2+P_{\rm near}}\right).
		\end{equation}
    }
\end{lemma}	

\begin{proof}
Based on \eqref{Eq:real_rete1} and \eqref{Eq:loss}, we have
		\begin{align}
			&\Delta_{R} = R^*_{\rm near}-R_{\rm near}\nn\\
			&=\log _2\left(1+\frac{P_{\rm near}g_{\rm near}}{\sigma^2}\right)-\log _2\left(1+\frac{P_{\rm near}g_{\rm near}}{P_{\rm far}g_{\rm near}f^2+\sigma^2}\right)\nn\\
			&=\log _2\left(1+\frac{P_{\rm near}P_{\rm far}g^2_{\rm near}f^2}{P_{\rm far}g_{\rm near}f^2\sigma^2+P_{\rm near}g_{\rm near}\sigma^2+\sigma^4}\right)\nn\\
			&\stackrel{(c)}{\le}\log _2\left(1+\frac{P_{\rm near}P_{\rm far}g^2_{\rm near}f^2}{P_{\rm far}g_{\rm near}f^2\sigma^2+P_{\rm near}g_{\rm near}\sigma^2}\right)\nn\\
			&=\log _2\left(1+\frac{P_{\rm near}g_{\rm near}}{\sigma^2}\cdot\frac{P_{\rm far}f^2}{P_{\rm far}f^2+P_{\rm near}}\right),
		\end{align}
	where $(c)$ is obtained by dropping the term $\sigma^4$, thus completing the proof.
\end{proof}
Lemma~\ref{Le:Rate_gap} is intuitively expected since a higher normalized interference power leads  to a larger rate loss $\Delta_{R}$. Combined with the results for the analysis of the interference power, it can be concluded that there is a larger rate loss when the NU and FU angles are very close. However, the analysis for the effects of number of BS antennas and NU distance are not straightforward. For example, when $N$ decreases, it can be shown that $g_{\rm near}$ monotonically decreases, while $f^2$ generally increases, leading to a larger $\frac{P_{\rm far}f^2}{P_{\rm far}f^2+P_{\rm near}}$. Thus, we further provide numerical results in the next to examine these effects.
\begin{table}[t]
	\caption{{Simulation parameters}}
	\label{Table1}
	\centering
	\begin{tabular}{|c|c|}
		\hline
		{\bf{Parameter}} & {\bf{Value}} \\
		\hline
		Number of BS antennas & $N=256$ \\
		\hline
		Carrier frequency & $f=30$ GHz \\
		\hline
		Reference path-loss & $\beta=(\lambda/4\pi)^2=-62$ dB \\
		\hline
		Transmit power for NU  & $P_{\rm near}=20$ dBm \\
		\hline
		Transmit power for FU & $P_{\rm far}=30$ dBm \\
		\hline
		Noise power & $\sigma^2=-70$ dBm \\
		\hline
		
		Distance from BS to NU & $r=3$ m \\
		\hline
	\end{tabular}
\vspace{-12pt}
\end{table}
\vspace{-12pt}
\subsection{Numerical Results}
\vspace{-2pt}
Last, we show the effects of the four key parameters on the achievable rate of the NU in Figs. \ref{fig:Loss_f2_N}--\ref{fig:Loss_f2_r} with the simulation parameters shown in Table~\ref{Table1}. {To begin with, we plot in Figs. \ref{fig:Loss_f2_N} and \ref{fig:Loss_f2_psi} the results based on the used Fresnel approximation in \eqref{Eq:Th1} and its actual value in \eqref{Eq:New_f} versus $N$ and $\psi$. It is observed that the used approximation is highly accurate with the actual value under different $N$ and $\psi$.}
Second, it is observed from Fig. \ref{fig:Loss_f2_N} that the rate loss is monotonically decreasing with $N$, which agrees with the result in Lemma~\ref{Le:Rate_gap}. An interesting observation is that there still exists a large rate loss when the number of antennas becomes very large (i.e., very low interference power). This is because, with an increasing $N$, the interference power generally decreases, while the beamforming gain (i.e., $g_{\rm near}$ in \eqref{Eq:loss}) increases, thus leading to the performance gap between the achievable and ideal rates of NU. On the other hand, the rate loss fluctuates when the FU angle is in the neighborhood of the NU angle, and drastically decreases when the angle is larger than a threshold (see Figs. \ref{fig:Loss_f2_psi} and \ref{fig:Loss_f2_theta}). Next, in Fig.~\ref{fig:Loss_f2_r}, we observe that the rate loss first slightly fluctuates when $r$ is small, and sharply decreases when $r$ increases. Last,
it is observed from Figs. \ref{fig:Loss_f2_psi}--\ref{fig:Loss_f2_r} that the achievable rate of NU will eventually approach to the ideal rate when the rate loss is sufficiently small.

\begin{figure}[t]
	\centering
	\subfigure[$\theta=0.05$,~$\psi=0$,~$r=3$]{\includegraphics[width=4.6cm]{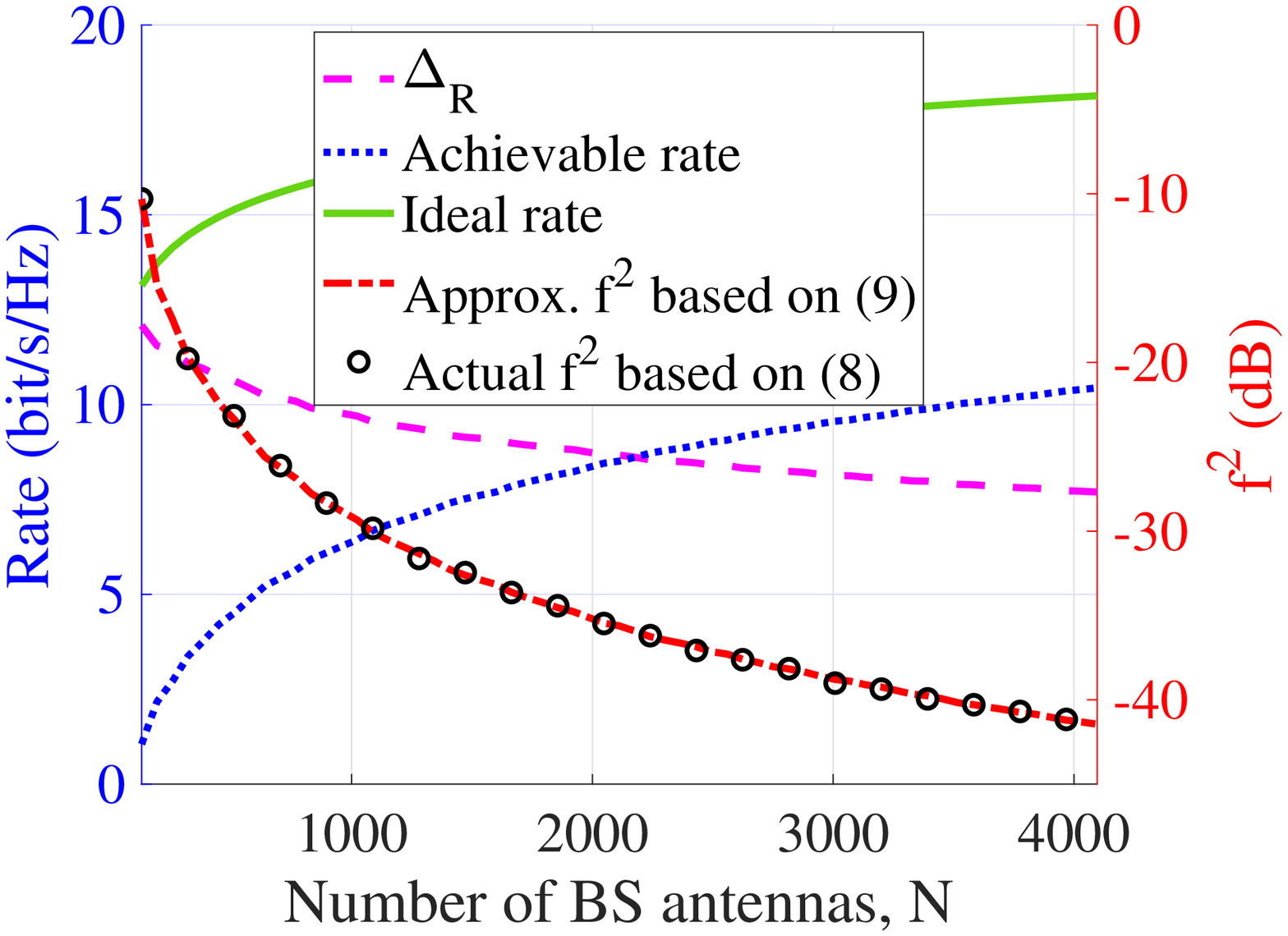}\label{fig:Loss_f2_N}}
	\hspace{-8pt}
	\subfigure[$\theta=0$,~$N=256$,~$r=3$]{\includegraphics[width=4.6cm]{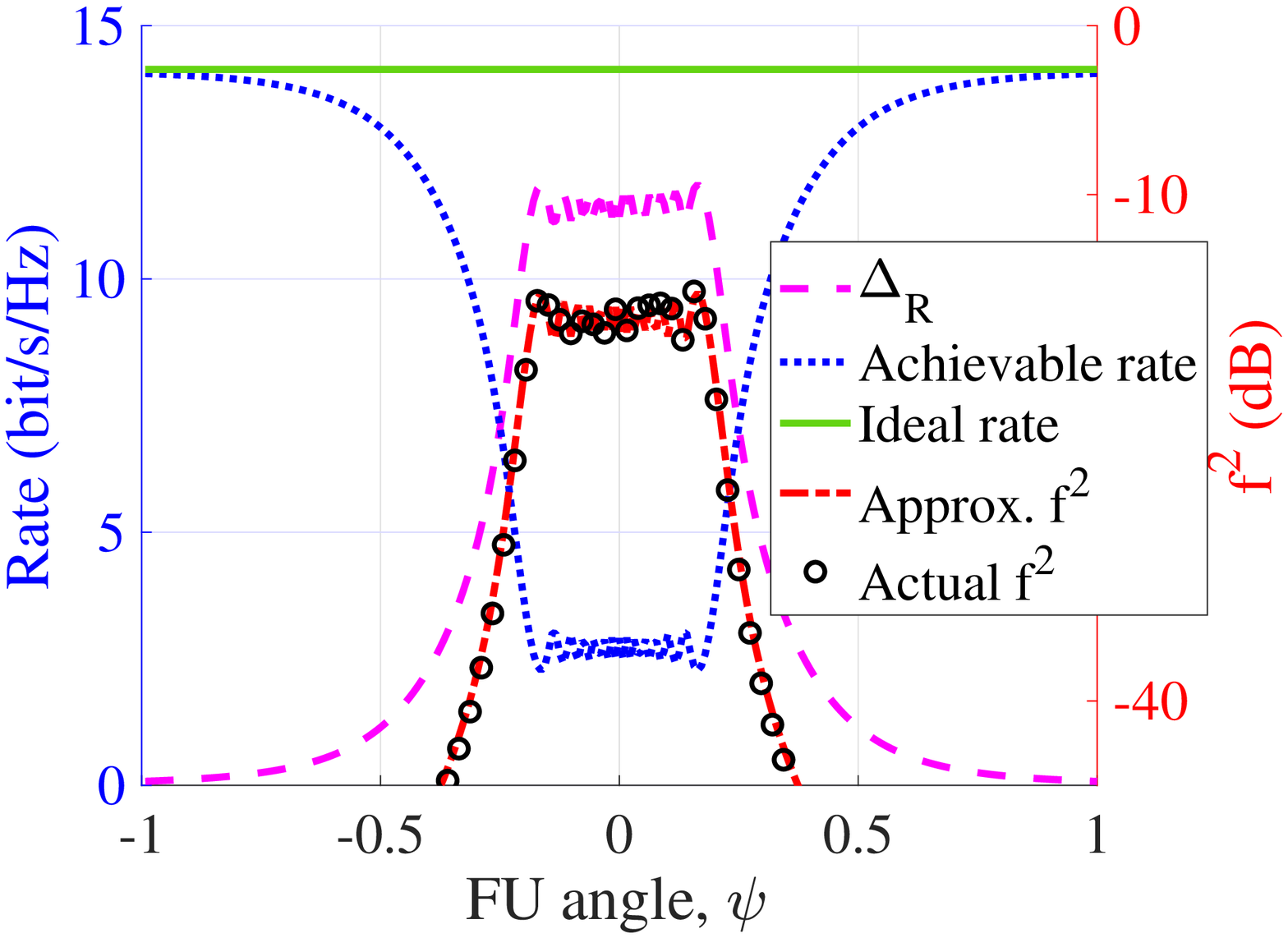}\label{fig:Loss_f2_psi}}
	\subfigure[$\psi=0$,~$N=256$,~$r=3$]{\includegraphics[width=4.6cm]{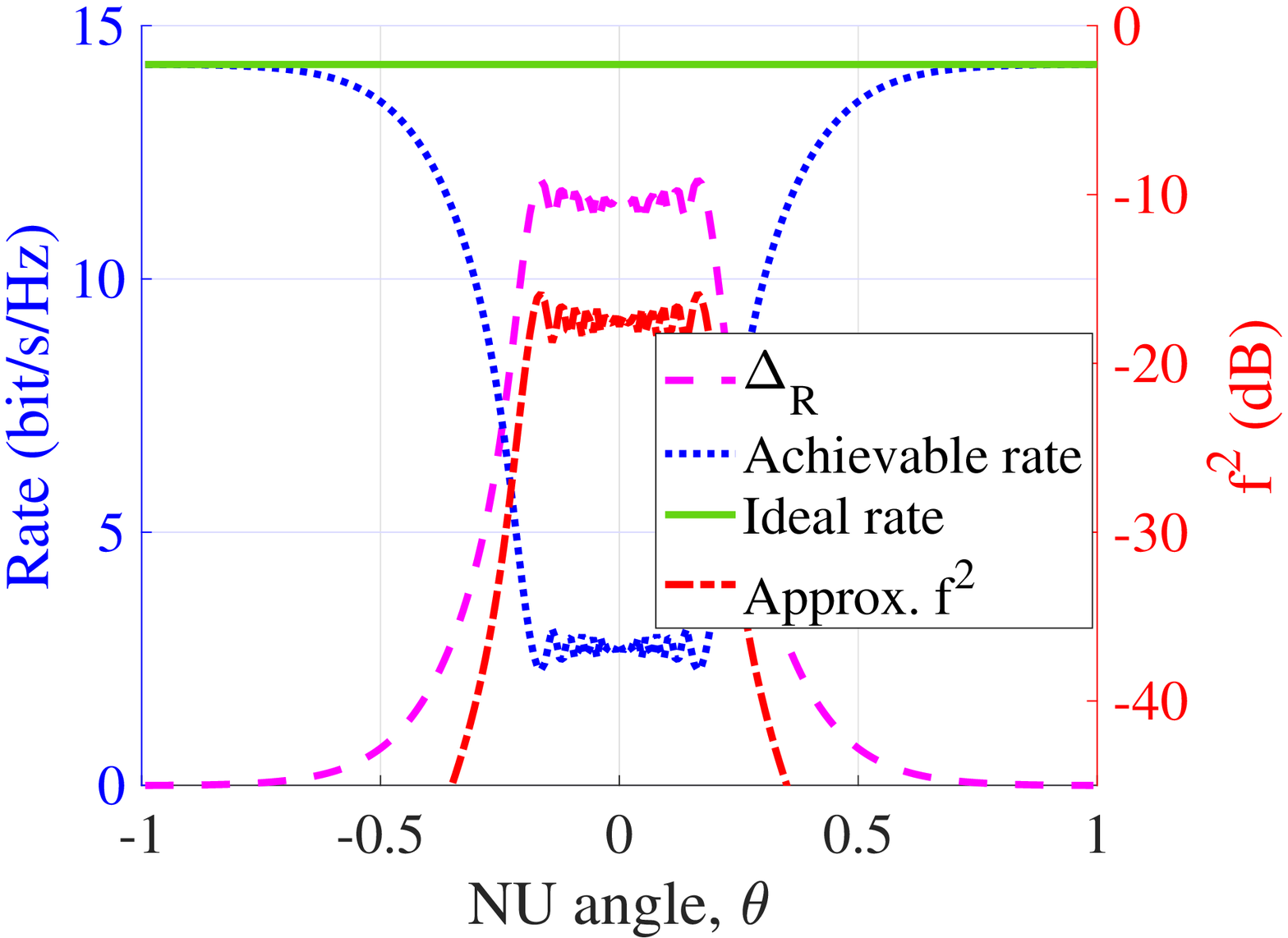}\label{fig:Loss_f2_theta}}
	\hspace{-8pt}
	\subfigure[$\theta=0.05$,~$\psi=0$,~$N=256$]{\includegraphics[width=4.6cm]{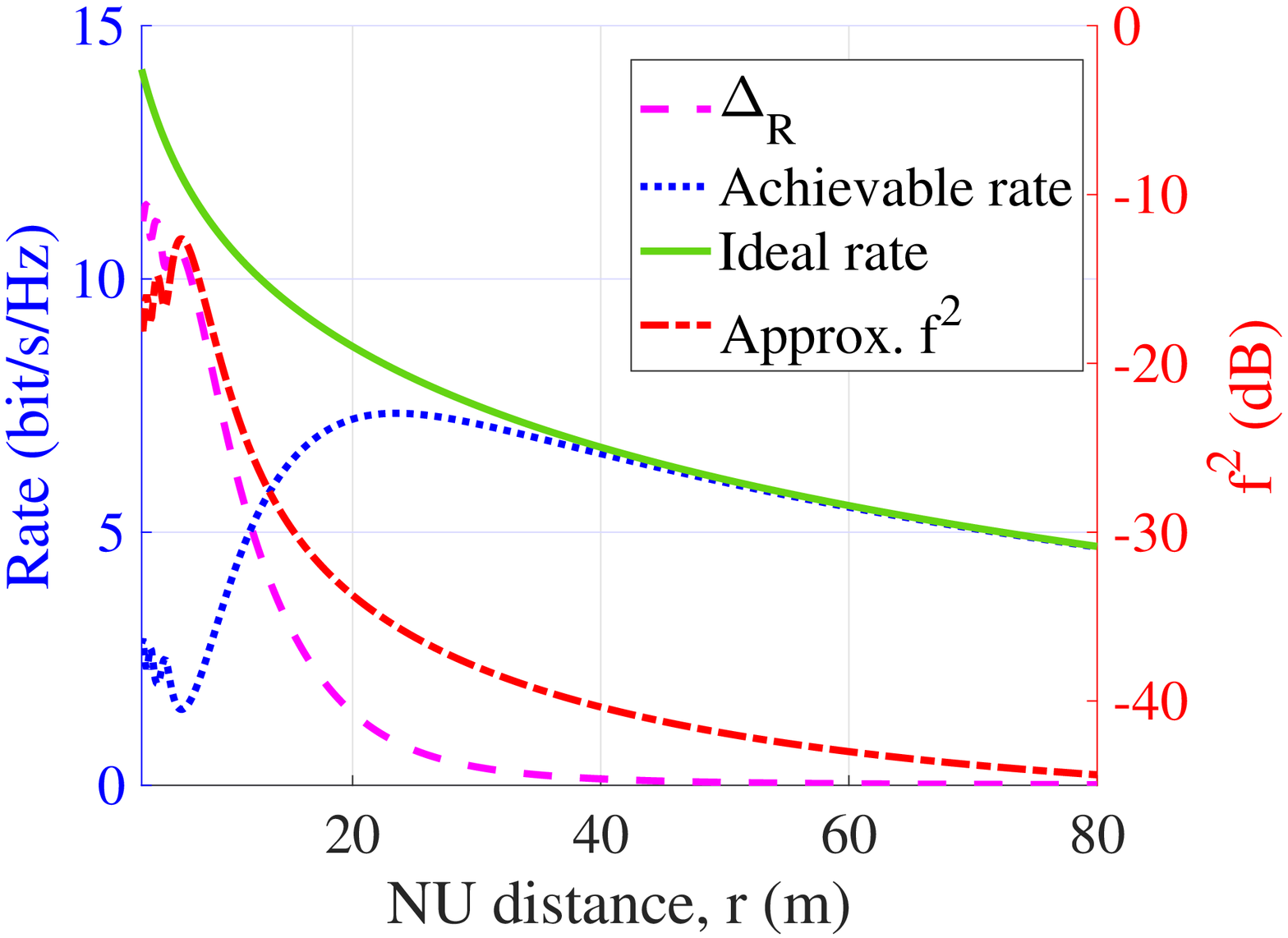}\label{fig:Loss_f2_r}}
	\caption{{Rate loss versus the number of BS antennas, the FU angle, the NU angle and distance.}}\label{fig:final}
	\vspace{-16pt}
\end{figure}
\section{Conclusions}

In this paper, we analyzed the inter-user interference in the new mixed near- and far-field communication scenario. Specifically, we first obtained a closed-form expression for the normalized interference power at the NU caused by the FU beam, based on which, the effects of the number of BS antennas, FU angle, NU angle and distance were analyzed. Moreover, the explicit rate-loss expression caused by the FU interference was obtained, which was verified by numerical results. In the future, this work can be extended in several directions. For example, it is interesting to consider different precoding designs (e.g., fully-digital and hybrid architectures), and multi-access techniques (e.g., non-orthogonal multiple access (NOMA)) to suppress inter-user interference.

\bibliographystyle{IEEEtran}
\vspace{-12pt}
\bibliography{IEEEabrv,Ref}

\begin{thebibliography}{10}
\providecommand{\url}[1]{#1}
\csname url@samestyle\endcsname
\providecommand{\newblock}{\relax}
\providecommand{\bibinfo}[2]{#2}
\providecommand{\BIBentrySTDinterwordspacing}{\spaceskip=0pt\relax}
\providecommand{\BIBentryALTinterwordstretchfactor}{4}
\providecommand{\BIBentryALTinterwordspacing}{\spaceskip=\fontdimen2\font plus
\BIBentryALTinterwordstretchfactor\fontdimen3\font minus
  \fontdimen4\font\relax}
\providecommand{\BIBforeignlanguage}[2]{{%
\expandafter\ifx\csname l@#1\endcsname\relax
\typeout{** WARNING: IEEEtran.bst: No hyphenation pattern has been}%
\typeout{** loaded for the language `#1'. Using the pattern for}%
\typeout{** the default language instead.}%
\else
\language=\csname l@#1\endcsname
\fi
#2}}
\providecommand{\BIBdecl}{\relax}
\BIBdecl

\bibitem{9903389}
M.~Cui, Z.~Wu, Y.~Lu, X.~Wei, and L.~Dai, ``Near-field communications for {6G}:
  Fundamentals, challenges, potentials, and future directions,'' \emph{IEEE
  Commun. Mag.\emph{,} early access}, 2022.

\bibitem{9326394}
Q.~Wu, S.~Zhang, B.~Zheng, C.~You, and R.~Zhang, ``Intelligent reflecting
  surface-aided wireless communications: A tutorial,'' \emph{IEEE Trans.
  Commun.}, vol.~69, no.~5, pp. 3313--3351, 2021.

\bibitem{9620081}
H.~Lu and Y.~Zeng, ``Near-field modeling and performance analysis for
  multi-user extremely large-scale {MIMO} communication,'' \emph{IEEE Commu.
  Lett.}, vol.~26, no.~2, pp. 277--281, 2022.

\bibitem{zhang2022fast}
Y.~Zhang, X.~Wu, and C.~You, ``Fast near-field beam training for extremely
  large-scale array,'' \emph{IEEE Wireless Commun. Lett.}, vol.~11, no.~12, pp.
  2625--2629, 2022.

\bibitem{wu2022multiple}
Z.~Wu, M.~Cui, and L.~Dai, ``Multiple access for near-field communications:
  {SDMA} or {LDMA}?'' \emph{arXiv preprint arXiv:2208.06349}, 2022.

\bibitem{7093168}
C.~Sun, X.~Gao, S.~Jin, M.~Matthaiou, Z.~Ding, and C.~Xiao, ``Beam division
  multiple access transmission for massive {MIMO} communications,'' \emph{IEEE
  Trans. Commun.}, vol.~63, no.~6, pp. 2170--2184, 2015.

\bibitem{9738442}
H.~Zhang, N.~Shlezinger, F.~Guidi, D.~Dardari, M.~F. Imani, and Y.~C. Eldar,
  ``Beam focusing for near-field multiuser {MIMO} communications,'' \emph{IEEE
  Trans. Wireless Commun.}, vol.~21, no.~9, pp. 7476--7490, 2022.

\bibitem{9129778}
C.~You, B.~Zheng, and R.~Zhang, ``Fast beam training for {IRS}-assisted
  multiuser communications,'' \emph{IEEE Wireless Commun. Lett.}, vol.~9,
  no.~11, pp. 1845--1849, 2020.

\bibitem{9723331}
E.~Bj{\"o}rnson, {\"O}.~T. Demir, and L.~Sanguinetti, ``A primer on near-field
  beamforming for arrays and reconfigurable intelligent surfaces,'' in
  \emph{2021 55th Asilomar Conference on Signals, Systems, and
  Computers}.\hskip 1em plus 0.5em minus 0.4em\relax IEEE, 2021, pp. 105--112.

\bibitem{luo2022beam}
H.~Luo and F.~Gao, ``Beam squint assisted user localization in near-field
  communications systems,'' \emph{arXiv preprint arXiv:2205.11392}, 2022.

\bibitem{deshpande2022wideband}
N.~Deshpande, M.~R. Castellanos, S.~R. Khosravirad, J.~Du, H.~Viswanathan, and
  R.~W. Heath~Jr, ``A wideband generalization of the near-field region for
  extremely large phased-arrays,'' \emph{arXiv preprint arXiv:2206.14323},
  2022.

\bibitem{9693928}
M.~Cui and L.~Dai, ``Channel estimation for extremely large-scale {MIMO}:
  Far-field or near-field?'' \emph{IEEE Trans. Commun.}, vol.~70, no.~4, pp.
  2663--2677, 2022.

\end{thebibliography}

\end{document}